\documentclass[11pt]{article}
\usepackage{makeidx}  
\usepackage{amsfonts,amsmath,amssymb,url,amsthm}
\usepackage{epsfig}
\usepackage{times}
\usepackage{graphics}
\usepackage{authblk}

\usepackage{color}
\usepackage{multicol}

\def\MIS{MIS}
\def\spec{\infty}

\def\Amax{AMAX}
\def\Bmax{BMAX}
\def\Cmax{CMAX}

\newcommand{\remove}[1]{}
\newcommand{\Active}{{{\sf Influenced}}}
\newcommand{\A}[1]{\Active[#1]}
\newcommand{\N}{{\mathbb{N}}}

\renewcommand{\tau}{\rho}

\def\jumpback{}

\newtheorem{theorem}{Theorem}
\newtheorem{lemma}{Lemma}
\newtheorem{definition}{Definition}
\newtheorem{remark}{Remark}

\begin{document}
\title{How to go Viral: Cheaply and Quickly}

\author[1]{F. Cicalese}

\author[2]{G. Cordasco}

\author[1]{L. Gargano}

\author[3]{M. Milani\v{c}}

\author[4]{J. Peters}

\author[1]{U. Vaccaro}

\affil[1]{Dept.~ of Computer Science, University of Salerno, Italy, {\texttt{\{cicalese,lg,uv\}@dia.unisa.it}}}
\affil[2]{Dept.~of Psychology, Second University of Naples, Italy, {\texttt {gennaro.cordasco@unina2.it}}}
\affil[3]{University of Primorska, UP IAM and UP FAMNIT, 
SI 6000 Koper, Slovenia,  \texttt{martin.milanic@upr.si}}
\affil[4]	{School of Computing Science, Simon Fraser University, Canada,   \texttt{peters@cs.sfu.ca}}


\maketitle              

\begin{abstract}
Given a social network represented by a graph $G$, 
we consider the problem of finding a bounded cardinality 
set of nodes $S$ with   the property  that  the influence spreading from  $S$ 
in $G$ is as large as possible. 
The dynamics that govern the spread of influence is the following:
 initially only  elements in $S$ are influenced;
subsequently at each round, the set of influenced  elements is
augmented by all nodes in the network that have a sufficiently large number of
already influenced  neighbors.
While it is known that the general problem is hard to solve --- even  in the approximate sense ---
we present exact polynomial time algorithms for trees, paths, cycles,
and complete graphs. \jumpback
\end{abstract}

\section{The Motivations}
Gaming giant FONY\textsuperscript{\textregistered} 
is about to launch its brand new  console  PlayForFUN-7\textsuperscript{\textregistered},
and intends  to maximize  the adoption of the  new product through 
a massive viral marketing  campaign,   
exploiting the human tendency to conform \cite{Asch56}.

This tendency occurs for three reasons: 
a) the basic human need to be liked and accepted  by others \cite{Baum+}; 
b) the belief that others, especially a majority group, 
have more accurate and trustworthy information than the individual \cite{S}; 
c) the ``direct-benefit'' effect,  implying that
an individual obtains an  explicit benefit when 
he/she aligns his/her behavior with the behavior of others (e.g., \cite{EK}, Ch.~17). 


In the case in point, argument c)  is supported by the fact that  
each player  who buys the PlayForFUN-7 console 
will be able to play online with all of the people 
who already have bought  the same console.
Indeed, the (possible) success of an on-line gaming service comes from its large number of users;
if this service had no members, there would be no point to anyone signing up for it. 
But as people begin using the service, the benefit for more people to sign up 
increases due to the increasing opportunities to play games with others online.
This motivates more people to sign up for the service which further increases the benefit.



FONY  is also aware  that the  much-feared  competitor Nanosoft\textsuperscript{\textregistered}  
will soon start to flood the market with  a very similar product: FUNBox-14. 
For this reason, it is crucial  to quickly spread the awareness of the new console PlayForFUN-7
to the  whole market of potential customers.

The CEO of FONY  enthusiastically embraced the idea of a viral marketing  
campaign\footnote{``If politicians can sell their stuff through a viral marketing campaign 
\cite{Bo+,LKLG,T}, then why not us?'', an unconfirmed source claims the CEO said.},
and instructed the FONY Marketing Division to plan a viral marketing campaign with the following requirements:
1) an initial set of influential people should be targeted and  receive a complimentary personalized
PlayForFUN-7 station (because of budget restrictions, this set is required to be \emph{small});
2) the group of influential people must be judiciously chosen so as to \emph{maximize} the spread of 
influence within the set of potential PlayForFUN-7 buyers;
3) the spread of influence must happen \emph{quickly}.

To comply with the CEO \emph{desiderata},  FONY Marketing Division analyzed
the behavior of  players  in the network  during the past few years 
(i.e., when players bought the latest console, 
how many games they bought, how many links/friends they have in 
the network, and how long they play on average every week).
On the basis of this analysis, an  estimate of each player's  tendency to conform
was made, and the following mathematical model was put forward.
The network of players is represented by a graph $G = (V,E)$,
where $V$ is the set of players, and there is an edge between two players
if those two players are friends in the network.
The individual's tendency to conform is quantified  by a function 
$t: V \longrightarrow \N = \{0,1,2,\ldots\}$,
with easy-to-convince players having ``low'' $t(\cdot)$ values, and
hard-to-convince players having ``high'' $t(\cdot)$ values.
If $S\subseteq V$ is any initial set of targeted people (\emph{target set}), then 
an {\em influence spreading process in $G$}, starting at $S$,
is a sequence of node  subsets 
$\A{S,0} \subseteq \A{S,1} \subseteq \ldots
\subseteq \A{S,\tau} \subseteq \ldots \subseteq V,$ such that
\begin{eqnarray*}
\A{S,0} &=& S \\ \mbox{and  for all }\tau > 0,& &\\
\A{S,\tau} &\!\!\!=\!\!\!& \A{S,\tau{-}1}\cup \Big\{u : \big|N(u){\cap} \A{S,\tau{-}1}\big|\ge t(u)\Big\},
\end{eqnarray*}
where $N(u)$ is the set of neighbors of $u$. 
In  words, an individual $v$ becomes influenced 
if  the  number of his influenced friends is at least 
its threshold $t(v)$.
It will be said  that $v$ is influenced {\em within} round $\tau$ if $v \in  \A{S,\tau}$;
 $v$ is influenced {\em at} round $\tau>0$ if $v \in  \A{S,\tau}\setminus \A{S,\tau - 1}$.
 
Using this terminology and notation, we can formally state the original problem as:

\medskip

\noindent
{\sc $(\lambda, \beta)$-Maximally Influencing Set ($(\lambda, \beta)$-MIS)}.\\
{\bf Instance:} A graph $G=(V,E)$, thresholds $t:V\longrightarrow \mathbb{N}$, 
a latency bound $\lambda\in \N$ and a budget $\beta \in \N$.\\
{\bf Question:} Find a set $S\subseteq V$  such that $|S|\leq \beta$ and $|\A{S,\lambda}|$ 
is as large as possible.

\medskip

\remove{
Notice that in the above problem we may assume without loss of generality
that $0\leq t(u)\leq d(u)+1$ holds for all nodes $u\in V$ (otherwise, we can set
$t(u)=d(u)+1$ for every node $u$ with threshold exceeding its degree plus one without
changing the problem).
}

\section{The Context}

It did not spoil the fun(!)\ of FONY Marketing Division to learn that 
(variants of) the $(\lambda, \beta)$-MIS problem have already been studied in
the scientific literature. 
We shall limit ourselves here
to discussing the work that is most directly related to ours, and refer the reader to the 
 monographs \cite{CLC,EK} for an excellent overview of the area.
We just mention that our results also seem to be  relevant to other  areas,
like dynamic monopolies \cite{FKRRS-2003,Peleg-02} for instance.

The first authors to study the spread of influence in networks
from an algorithmic point of view were Kempe \emph{et al.} \cite{KKT-03,KKT-05}.
However, they were mostly interested in networks with  randomly chosen thresholds.
Chen \cite{Chen-09} studied the following minimization problem:
given a graph $G$ and fixed thresholds $t(v)$, find
a  set of minimum size that eventually influences
all (or a fixed fraction of) nodes of $G$.
He proved  a  strong inapproximability result that makes unlikely the existence
of an  algorithm with  approximation factor better than  $O(2^{\log^{1-\epsilon }|V|})$.
Chen's result stimulated a series of papers  \cite{ABW-10,BHLM-11,BCNS,Centeno12,Chiang,Chopin-12,Chun,Chun2,C-OFKR,Ga+,Re,Za}, that isolated interesting cases 
in which the problem (and variants thereof) becomes tractable.

None of these papers considered 
the \emph{number of rounds} necessary for the spread of influence in the network. 
However, this is a relevant question for viral marketing in which
it is quite important to spread information quickly.
Indeed, research in Behavioural Economics 
shows that humans make decisions mostly on the basis of very recent events, 
even though they might remember much more \cite{Alba,Chen+}.
The only paper known to us that has studied the spread of influence in the same diffusion model
that we consider here, and with constraints on the number of rounds 
in which the process must be completed, 
 is \cite{CCGMV13}. How our results are related to \cite{CCGMV13}
 will be elucidated   in the next section. 
\remove{
It is equally important, before embarking on a possible
onerous investment,
to try  estimating the maximum amount of influence spread that can
be guaranteed within a certain amount of time (i.e, for
some $\lambda$ fixed in advance),
rather than simply knowing that eventually (but maybe too late)
the whole market might be covered. These considerations motivate our first generalization of
the problem, parametrized by the number of rounds $\lambda.$
The practical relevance of parameterizing the problem also
with  bounds on the initial budget or the final requirement should
be equally evident.}%
Finally, we
 point out that Chen's  \cite{Chen-09} inapproximability result
still holds for general graphs if the diffusion  process must end in a
bounded number of rounds.

\section{The Results}
Our main results are polynomial time algorithms to 
solve the $(\lambda,\beta)$-MIS problem on Trees, Paths, 
Cycles, and Complete
graphs, improving and extending some results from \cite{CCGMV13}.
In particular, the paper \cite{CCGMV13} put forward
an algorithmic framework to solve the $(\lambda,\beta)$-MIS problem (and related
ones), in graphs of bounded clique-width. When instantiated on trees,
the approach of \cite{CCGMV13} would give algorithms for the $(\lambda,\beta)$-MIS problem
with complexity that is
\emph{exponential} in the parameter $\lambda$, whereas our algorithm 
has complexity polynomial in all the relevant parameters (cf., Theorem \ref{theorem-tree}).
We should also remark that, in the very special  case $\lambda=1$ and thresholds $t(v)=1$, for each $v\in V$,
problems of influence diffusion reduce to well known domination problems in graphs (and variants thereof).
In particular, when $\lambda=1$
 and $t(v)=1$, for each $v\in V$,
our $(\lambda,\beta)$-\textsc{Maximally Influencing Set} problem reduces to the \textsc{Maximum Coverage}
problem considered in \cite{BGHHJK}.
Therefore, our results can also be seen as far-reaching generalizations of \cite{BGHHJK}.

\medskip

\section{$(\lambda,\beta)$-Maximally Influencing Set  on Trees}	\label{sec-trees}
In this section, we give an algorithm for the {\sc $(\lambda, \beta)$-Maximally Influencing Set} problem on trees.
Let $T = (V,E)$ be a tree, rooted at some node $r$.
Once such a rooting is fixed, for any node $v$, we  denote by $T(v)$ the  subtree rooted at $v$.
We will develop a dynamic programming algorithm 
that will
 prove the following theorem.
\begin{theorem}\label{theorem-tree}
The {\sc $(\lambda, \beta)$-Maximally Influencing Set} problem can be 
solved in  time \\$O(\min\{n\Delta^2\lambda^2\beta^3,$ $n^2\lambda^2\beta^3\})$ on  a tree with $n$ nodes and maximum degree $\Delta$.
\end{theorem}

 The rest of this section is devoted to the description and analysis of the
  algorithm that proves  Theorem \ref{theorem-tree}.
The algorithm   traverses the input  tree $T$  bottom up, in such a way that each node is considered 
after all its children have been processed. For each node $v$, the algorithm solves all possible   
{$(\lambda, b)$-MIS} problems on the subtree $T(v)$, for $b=0,1,\ldots, \beta$.  Moreover, in order to compute these values 
we will have to  consider not only the original threshold $t(v)$ of $v$, but also 
the decreased value $t(v)-1$ which we call the {\em residual threshold}.
In the following,  we assume without loss of generality that $0\leq t(u)\leq d(u)+1$ 
(where $d(u)$ denotes the degree of $u$) holds for all nodes $u\in V$ (otherwise, we can set
$t(u)=d(u)+1$ for every node $u$ with threshold exceeding its degree plus one without
changing the problem).
\begin{definition} 
For each node $v\in V$,  integers $b\geq 0$,   
$t\in \{t(v)-1,t(v) \}$, and $\tau \in\{0,1,\ldots,\lambda\}\cup \{\spec\}$, let us
denote by $\MIS[v,b,\tau,t]$   the maximum number of nodes that can be influenced in $T(v)$, within round $\lambda$, assuming that
\begin{itemize}
\item at most $b$ nodes among those in $T(v)$ belong to the target set;
\item the threshold of $v$ is $t$;
\item the parameter $\tau$ is such that
\begin{align}
 & 1) \mbox{ if $\tau=0$ then $v$ must belong to the target set,}\label{eq-case1}\\
 & 2) \mbox{ if $1\leq\tau\leq \lambda$ then $v$ is not in the target set and at least $t$ of its children are active  }\nonumber\\
     & \quad        \mbox{  within round $\tau-1$,}\label{eq-case2}\\
 & 3) \mbox{ if $\tau=\infty$  then $v$ is not influenced within round $\lambda$.}\label{eq-case3} 
\end{align}
\end{itemize}

We define $\MIS[v,b,\tau,t]= -\infty$ when 
any of the above constraints is not satisfiable. 
For instance, if $b=\tau=0$ we have\footnote{Since  $\tau=0$ then $v$ should belong to the target set, but 
this is not possible because the  budget is $0$.} $\MIS[v,0,0,t]= -\infty$. 

Denote  by $S(v,b,\tau,t)$ any target set attaining the value $\MIS[v,b,\tau,t]$.
\end{definition}
We notice that in the above definition  
 if $1\leq \tau\leq \lambda$ then,    the assumption that $v$ has threshold $t$ implies that 
 $v$ is influenced within round $\tau$
and is   able to influence its neighbors starting from round $\tau +1$.
The value
$\tau=\spec$ means that no condition are imposed on  
$v$: It could be influenced after round $\lambda$ or not influenced at all.
 In the sequel, $\tau=\spec$ will be used to 
ensure that $v$ will not contribute to the influence any neighbor (within  round $\lambda$). 
\begin{remark}
It is worthwhile mentioning that  $\MIS[v,b,\tau,t]$ is monotonically non-decreasing in $b$  and non-increasing in $t$.
However, $\MIS[v,b,\tau,t]$ is not necessarily  monotonic in $\tau$.
\end{remark}
The maximum number of nodes in $G$ that can be influenced within round $\lambda$ with
any (initial) target set of cardinality at most $\beta$
can be then obtained by computing
\begin{equation}\label{eq-mas}
\max_{\tau \in \{0,1,\ldots,\lambda,\spec\}} \ \MIS[r,\beta,\tau,t(r)].
\end{equation}
In order to obtain the value in (\ref{eq-mas}), we compute $\MIS[v,b,\tau,t]$ for each $v \in V,$  for each $b=0,1,\ldots,\beta$, for each $\tau\in\{0,1,\ldots, \lambda, \spec\}$,  and for $t \in \{t(v)-1,t(v)\}$.

%

We  proceed in a  bottom-up fashion on the tree, so that the computation of the various values $\MIS[v,b,\tau,t]$ for a node $v$ is done after all 
the values for $v$'s children are known.

For each leaf node $\ell$ we have
\begin{equation}\label{eq-casel}
\MIS[\ell,b,\tau,t] =  \begin{cases} 1 & \mbox{ if } (\tau=0 \mbox{ AND } b\geq1) \mbox{ OR } (t=0 \mbox{ AND }1\leq \tau \leq \lambda ) 
\\
0 & \mbox{ if } \tau=\infty
\\
-\infty &  \mbox{otherwise.} \end{cases} 
\end{equation}
Indeed, a leaf $\ell$ gets influenced, in the single node subtree $T(\ell)$, only when either   $\ell$ belongs to the target set ($\tau=0$) and the budget is  sufficiently large ($b\geq1$)  or  the threshold is zero (either $t=t(\ell)=0$ or $t=t(\ell)-1=0$) independently of the number of rounds.

For an internal node $v$, we show how to compute each  value $\MIS[v,b,\tau,t]$   in time $O(d(v)^2\lambda\beta^2)$.

We recall that when computing  a value $\MIS[v,b,\tau,t]$, we already have computed all the $\MIS[v_i,*,*,*]$ values for each child $v_i$ of $v$.\\
We  distinguish three cases  for the computation of $\MIS[v,b,\tau,t]$ 
 according to the value of $\tau$.

\medskip

\noindent \textbf{CASE 1:} $\tau=0$.
In this case we assume that $b\geq 1$ (otherwise $\MIS[v,0,0,t]=-\infty$). Moreover,   we know that $v\in S(v,b,0,t)$ hence the computation of 
$\MIS[v,b,0,t]$
must  consider all the possible ways in which the remaining budget $b-1$ can be partitioned among $v$'s children.
\begin{lemma}\label{lemma1}
It is possible to compute $\MIS[v,b,0,t]$, where $b\geq 1$,  in time $O(d \lambda b^2),$ where $d$ is the number of children of $v$.
\end{lemma}
\begin{proof}
Fix an ordering $v_1,v_2,\ldots, v_d$ of the children of node $v$.\\
For $i=1,\ldots, d$ and $j=0,\ldots, b-1$, let $\Amax_v[i,j]$ be the maximum number of nodes that can be influenced, within $\lambda$ rounds, in $T(v_1),T(v_2),\ldots,T(v_i)$ assuming that the target set contains  $v$ and at most $j$ nodes  among those in $T(v_1),T(v_2),\ldots,T(v_i)$.

By (\ref{eq-case1}) we have 
\begin{equation}\label{eq-amax}
\MIS[v,b,0,t]=1+\Amax_v[d,b-1].
\end{equation}
We now show how to compute $\Amax_v[d,b-1]$ by recursively computing the values $\Amax_v[i,j]$, for each $i=1,2,\ldots,d$ and  $j=0,1,\ldots,b-1$.

For $i=1$, we  assign all of the budget to $T(v_1)$ and 
$$\Amax_v[1,j]= 
\max_{\tau_1,t_1} \{\MIS[v_1,j,\tau_1,t_1]\},$$
where 
 $\tau_1 \in \{0,\ldots,\lambda,\spec\}$,
\quad  $t_1\in \{ t(v_1), t(v_1)-1\}$,  \quad and 
	\quad  if $t_1=t(v_1)-1$ then $  \tau_1 \geq  1.$

For $i>1$, we consider all possible ways of partitioning the budget  $j$ into two values $a$ and $j-a$, for each $0\leq a \leq j$. The budget $a$ is assigned to the first $i-1$ subtrees, while the  budget $j-a$ is assigned to $T(v_i)$. Hence,  
$$   \Amax_v[i,j]=\max_{0\leq a \leq j} \left \{ \Amax_v[i-1,a] + \max_{\tau_i,t_i} \{\MIS[v_i,j-a,\tau_i,t_i]\} \right \}$$ 
where 
$\tau_i \in \{0,\ldots,\lambda,\spec\},$
\quad  $t_i\in \{ t(v_i), t(v_i)-1\}$, \quad and
	\quad   if $t_i=t(v_i)-1$ then $  \tau_i \geq  1$.

The computation of  $\Amax_v$ comprises  $O(d  b)$ values and each one is computed recursively in time $O(\lambda b)$.  Hence 
we are able to compute it, and 
by (\ref{eq-amax}),  also $\MIS[v,b,0,t]$, in time $O(d \lambda b^2)$.
\end{proof}

\medskip

\noindent \textbf{CASE 2:} $1\leq\tau\leq\lambda$.
In this case     $v$ is not in the target set and  at round $\tau-1$ at least $t$ of its children  must be influenced.
The computation of a value $\MIS[v,b,\tau,t]$ must  consider all the possible ways in which the  budget $b$ can be partitioned among $v$'s children in such a way that at least $t$ of them are influenced within round $\tau-1$.

\begin{lemma}
For each $\tau=1,\ldots, \lambda$, it is possible to
 compute $\MIS[v,b,\tau,t]$ recursively in time $O(d^2\lambda b^2),$ where $d$ is the number of children of $v$.
\end{lemma}
\begin{proof}
Fix any ordering $v_1,v_2,\ldots, v_d$ of the children of the node $v$.\\
We first define the values $\Bmax_{v,\tau}[i,j,k]$, for   $i=1,\ldots, d$, $j=0,\ldots, b$, and $k=0,\ldots, t$.\\
If $i\geq k$, we define $\Bmax_{v,\tau}[i,j,k]$ to be  the maximum number of nodes that can be influenced, within $\lambda$ rounds, in the subtrees $T(v_1),T(v_2),\ldots,T(v_i)$ assuming that
\begin{itemize}
\item $v$ is influenced within round $\tau;$ 
\item at most $j$ nodes among those in $T(v_1),T(v_2),\ldots,T(v_i)$  belong to the target set;
\item  at least $k$  among   $v_1,v_2,\ldots,v_i$, will be influenced within round $\tau-1$.
\end{itemize}
We define $\Bmax_{v,\tau}[i,j,k]= -\infty$ when the above constraints are not satisfiable. For instance, if $i<k$ we have $\Bmax_{v,\tau}[i,j,k]= -\infty$. 
\medskip

By (\ref{eq-case2}) and by the definition of $\Bmax$,   we have    
\begin{equation}\label{eq-bmax}\MIS[v,b,\tau,t]=1+\Bmax_{v,\tau}[d,b,t].\end{equation} 
We can  compute $\Bmax_{v,\tau}[d,b,t]$ by recursively computing the values of $\Bmax_{v,\tau}[i,j,k]$ for each $i=1,2,\ldots,d,$ for each $j=0,1,\ldots,b,$ and for each $k=0,1,\ldots,t,$ as follows.

\def\taum{\delta}

For $i=1$, we  have to assign all the budget $j$ to the first subtree of $v$. Moreover, if $k=1$, then by definition $v_1$ has to be influenced before round $\tau$ and consequently we can not use  threshold $t(v_1)-1$ (which assumes  that  $v$ contributes to the influence of $v_i$). Hence, we have 
\begin{equation}\Bmax_{v,\tau}[1,j,k]= \begin{cases}   
\max_{\tau_1,t_1} \{\MIS[v_1,j,\tau_1,t_1]\}, & \mbox{if } k=0 \\	
\max_{\taum}\{\MIS[v_1,j,\taum,t(v_1)]\}, & \mbox{if } k=1 \\
 -\infty, & \mbox{otherwise,} \end{cases} \label{eqBmax}\end{equation}
\jumpback
\jumpback
where 
\begin{itemize}
\item $\tau_1 \in \{0,\ldots,\lambda,\spec\}$
\item $t_1\in \{ t(v_1), t(v_1)-1\}$
	\item if $t_1=t(v_1)-1$ then $  \tau_1 \geq  \tau+1$ 
\item $\taum\in \{0,\ldots, \tau-1\}$.
\end{itemize}

The third constraint ensures that we can use a reduced threshold on $v_1$ only after the
father $v$ has been influenced.

To show   the correctness of equation (\ref{eqBmax}), one can (easily)
 check that, for $k <2$, any target set solution $S$ that maximizes the value on the 
 left side of the equation is also a feasible solution for the value on the right, and vice versa.
 
For $i>1$, as in the preceding lemma, we consider all possible ways of partitioning the budget  $j$ into two values $a$ and $j-a$. The budget $a$ is assigned to the first $i-1$ subtrees, while the remaining budget $j-a$ is assigned to $T(v_i)$. Moreover, in order to ensure that  at least $k$ children of $v$, among children  $v_1,v_2,\ldots,v_i$, will be influenced before round $\tau$, there are two cases to consider: a) the $k$ children that are influenced before round $\tau$ are among the first $i-1$ children of $v$. In this case $v_i$ can be influenced at any round and can use a reduced threshold; b) only $k-1$ children among nodes  $v_1,v_2,\ldots,v_{i-1}$ are influenced before round $\tau$ and consequently $v_i$ has to be influenced before round $\tau$ and  cannot use a reduced threshold.   Formally, we prove that
\begin{eqnarray}\label{eqBmax2}
\nonumber \Bmax_{v,\tau}&[i,j,k]  {=}\max \Big\{ \max_{\substack{0\leq a\leq j\\ \tau_i, t_i}}  (\Bmax_{v,\tau}[i{-}1,a,k] + \MIS[v_i,j{-}a,\tau_i,t_i]),\\
                                  & \max_{\substack{0\leq a\leq j\\ \taum}} (\Bmax_{v,\tau}[i{-}1,a,k{-}1]+\MIS[v_i,j{-}a,\taum,t(v_i)]) \Big\}
\end{eqnarray}
\jumpback
where 
\begin{itemize}
\item $\tau_i \in \{0,\ldots,\lambda,\spec\}$
\item $t_i\in \{ t(v_i), t(v_i)-1\}$
\item if $t_i=t(v_i)-1$ then $  \tau_i \geq  \tau+1$ 
\item $\taum\in \{0,\ldots, \tau-1\}$.
\end{itemize}

In the following we show the correctness of equation (\ref{eqBmax2}).
First we show that 
\begin{equation*}
\begin{split}
\Bmax_{v,\tau}[i,j,k]\leq&\max \Big\{  \max_{\substack{0\leq a\leq j\\ \tau_i, t_i}} (\Bmax_{v,\tau}[i{-}1,a,k] + \MIS[v_i,j{-}a,\tau_i,t_i]),\\
                                  & \max_{\substack{0\leq a\leq j\\ \taum}} (\Bmax_{v,\tau}[i{-}1,a,k{-}1]+\MIS[v_i,j{-}a,\taum,t(v_i)]) \Big\}
\end{split}
\end{equation*}

\remove{
$$\Bmax_{v,\tau}[i,j,k]\leq \max_{0\leq a \leq j}\left .\begin{cases}  \Bmax_{v,\tau}[i{-}1,a,k] + \max_{\tau_i,t_i} \{\MIS[v_i,j{-}a,\tau_i,t_i]\},\\ \Bmax_{v,\tau}[i{-}1,a,k{-}1] + \max_{\taum}\{\MIS[v_i,j{-}a,\taum,t(v_i)]\}   \end{cases}\right .$$
}

 Let $S\subseteq \bigcup_{z=1}^{i}T(v_z)$ be a feasible target set solution that maximizes
the number of nodes that can be influenced, within $\lambda$ rounds, in the subtrees $T(v_1),T(v_2),\ldots,T(v_i)$ and satisfies the constraints defined in the definition of $\Bmax_{v,\tau}[i,j,k]$. Hence $|S| \leq j$. We can partition $S$ into two sets $S_a$, where $|S_a|\leq a$, and $S_b$ ($|S_b|\leq j-a$) in such a way that $S_a \subseteq \bigcup_{z=1}^{i-1}T(v_z)$ while $S_b \subseteq T(v_i)$. Since $S$ satisfies the constraints defined in the definition of $\Bmax_{v,\tau}[i,j,k]$, we have that, starting with $S$, at least $k$ children of $v$, among children  $v_1,v_2,\ldots,v_i$, will be influenced before round $\tau$. Hence,  starting with $S_a$, at least $k-1$ children of $v$, among children  $v_1,v_2,\ldots,v_{i-1}$, will be influenced before round $\tau$. We distinguish  two cases:
\begin{itemize}
	\item If $S_a$ influences $k-1$ children of $v$, among children  $v_1,v_2,\ldots,v_{i-1}$, before round $\tau$, then we have that $S_b$ must also influence  $v_i$ before round $\tau$.  Hence $S_a$ is a feasible solution for $\Bmax_{v,\tau}[i{-}1,a,k{-}1]$ and $S_b$ is a feasible solution for	\\ $\max_{\taum}\{\MIS[v_i,j{-}a,\taum,t(v_i)]\}$.

\item On the other hand when $S_a$ influences at least $k$ children of $v$, among children  $v_1,v_2,\ldots,v_{i-1}$, before round $\tau$ then $S_a$ is a feasible solution for $\Bmax_{v,\tau}[i{-}1,a,k]$ and $S_b$ is a feasible solution for $\max_{\tau_i,t_i} \{\MIS[v_i,j{-}a,\tau_i,t_i]\}$. 
\end{itemize}
In either case we have that the solution $S$ is also a solution for the right side of the equation. 
%
Perfectly similar reasoning can be used to show that 
\begin{equation*}
\begin{split}
\Bmax_{v,\tau}[i,j,k]\geq&\max \Big\{  \max_{\substack{0\leq a\leq j\\ \tau_i, t_i}} (\Bmax_{v,\tau}[i{-}1,a,k] + \MIS[v_i,j{-}a,\tau_i,t_i]),\\
                                  & \max_{\substack{0\leq a\leq j\\ \taum}} (\Bmax_{v,\tau}[i{-}1,a,k{-}1]+\MIS[v_i,j{-}a,\taum,t(v_i)]) \Big\}
\end{split}
\end{equation*}

\remove{
\jumpback  $$\Bmax_{v,\tau}[i,j,k]\geq \max_{0\leq a \leq j}\left .\begin{cases}  \Bmax_{v,\tau}[i{-}1,a,k] + \max_{\tau_i,t_i} \{\MIS[v_i,j{-}a,\tau_i,t_i]\},\\ \Bmax_{v,\tau}[i{-}1,a,k{-}1] + \max_{\taum}\{\MIS[v_i,j{-}a,\taum,t(v_i)]\}   \end{cases}\right .$$
}
and hence equation (\ref{eqBmax2}) is proved.

The computation of $\Bmax_{v,\tau}$ comprises $O(d^2  b)$ values (recall that $t\leq d+2$) and each one is computed recursively in time $O(\lambda b)$. Hence we are able to compute it, and 
by (\ref{eq-bmax}),  also $\MIS[v,b,\tau,t]$, in time $O(d^2 \lambda b^2 )$.
\end{proof}

\medskip

\noindent \textbf{CASE 3:} $\tau=\spec$.
In this case we only have to consider the original threshold $t(v_i)$ for each child $v_i$ of $v$.
Moreover, we must  consider all the possible ways in which the  budget $b$ can be partitioned among $v$'s children.

\begin{lemma}
It is possible to compute $\MIS[v,b,\spec,t]$  in time $O(d\lambda b^2),$ where $d$ is the number of children of $v$.
\end{lemma}
\begin{proof}
Fix any ordering $v_1,v_2,\ldots, v_d$ of the children of the node $v$.\\
For $i=1,\ldots, d$ and $j=0,\ldots, b$, let $\Cmax_v[i,j]$ be the maximum number of nodes that can be influenced, within $\lambda$ rounds, in $T(v_1),T(v_2),\ldots,T(v_i)$ assuming that 
\begin{itemize}
\item $v$ will not be influenced within $\lambda$ rounds and 
\item at most $j$ nodes, among nodes in $T(v_1),T(v_2),\ldots,T(v_i)$,  belong to the target set.
\end{itemize}
By (\ref{eq-case3}) and by the definition of $\Cmax$,   we have  \jumpback    
\begin{equation} \label{eq-cmaxmis}\MIS[v,b,\spec,t]=\Cmax_v[d,b].\end{equation} 
We can compute $\Cmax_v[d,b]$ by  recursively computing the values  $\Cmax_v[i,j]$ for each $i=1,2,\ldots,d$ and for each $j=0,1,\ldots,b,$  as follows.
\\
For $i=1$, we can assign all of the budget to the first subtree of $v$ and we have \jumpback$$\Cmax_v[1,j]= 
\max_{\tau_1} \{\MIS[v_1,j,\tau_1,t(v_1)]\}$$ 
where $\tau_1 \in \{0,\ldots,\lambda,\spec\}.$
\\
For $i>1$,  we consider all possible ways of partitioning the budget  $j$ into two values $a$ and $j-a$, for each $0\leq a \leq j.$ The budget $a$ is assigned to the first $i-1$ subtrees, while the remaining budget $j-a$ is assigned to $T(v_i)$. Hence,  the following holds: \jumpback 
$$   \Cmax_v[i,j]=\max_{0\leq a \leq j} \left \{ \Cmax_v[i-1,a] + \max_{\tau_i} \{\MIS[v_i,j-a,\tau_i,t(v_i)]\} \right \}$$ 
where $\tau_i \in \{0,\ldots,\lambda,\spec\}.$


The computation of $\Cmax_v$ comprises $O(d b)$ values and each one is computed recursively in time $O(\lambda b)$. Hence,
by (\ref{eq-cmaxmis}),   we are able to compute $\MIS[v,b,\spec,t]$ in time $O(d\lambda b^2)$.
\end{proof}

Thanks to the three lemmas above we have that for each node $v \in V,$  for each $b=0,1,\ldots,\beta$, for each $\tau=0,1,\ldots, \lambda, \spec$, and for $t \in \{t(v)-1,t(v)\}$, $\MIS[v,b,\tau,t]$ can be computed recursively in time $O(d(v)^2\lambda\beta^2)$. Hence, the value
$$\max_{\tau \in \{0,1,\ldots,\lambda,\spec\}} \ \MIS[r,\beta,\tau,t(r)]$$ can be computed in time
\jumpback $$\sum_{v \in V} O(d(v)^2\lambda\beta^2) {\times} O(\lambda \beta)=O(\lambda^2\beta^3)\times\sum_{v \in V} O(d(v)^2)=O(\min\{n\Delta^2\lambda^2\beta^3,n^2\lambda^2\beta^3\}),$$
where $\Delta$ is the maximum node degree.
Standard backtracking techniques can be used to compute a target set of cardinality at most $\beta$ that
influences this maximum number of nodes in the same $O(\min\{n\Delta^2\lambda^2\beta^3,n^2\lambda^2\beta^3\})$ time.
This proves Theorem \ref{theorem-tree}.

\newcommand\lA[1]{p_A(#1)}

\newcommand\lS[1]{p_S(#1)}

\section{$(\lambda,\beta)$-Maximally Influencing Set on Paths,   Cycles, and Complete Graphs}	 

The results of Section \ref{sec-trees} obviously include paths. 
However, we are able to significantly improve on the computation time for paths.

Let $P_n=(V,E)$  be a path on $n$ nodes $v_1,v_2,\ldots, v_{n}$, 
and  edges  $(v_i,v_{i+1}$), for  $i= 1,\ldots, n-1$.
Moreover, we denote by $C_n$ the cycle on $n$ nodes that consists  
of the path $P_n$ augmented with the edge $(v_1,v_n)$.
In the following, we assume that $ 1 \leq t(i) \leq 3$, for  
$i= 1,\ldots, n.$ Indeed, paths with $0$-threshold nodes can 
be dealt with by removing up to $\lambda$ $1$-threshold nodes 
on the two sides of each $0$-threshold node. In case we remove 
strictly less than $\lambda$ nodes, we can reduce by $1$ the 
threshold of the first node that is not removed (which must have 
threshold greater than $1$). The path gets split into several subpaths, 
but the construction we provide below still works (up to taking care of boundary conditions).

\begin{theorem}\label{teo-path}
The {$(\lambda, \beta)$-\textsc{Maximally Influencing Set}} problem can be solved in time $O(n\beta\lambda)$ on a path $P_n$.
\end{theorem}
%
%
\begin{proof}{} {\bf (Sketch.)}
For $i=1,2, \ldots n,$ let $r(i)$ be the number of consecutive nodes having threshold $1$ on the right of  node $v_i$, that is,  $r(i)$ is the largest integer such that $i+r(i) \leq n$ and $t(v_{i+1})=t(v_{i+2})=\ldots=t(v_{i+r(i)})=1$. 
Analogously we define $l(i)$ as the largest  integer such that $i-l(i)\geq 1$ and $t(v_{i-1})=t(v_{i-2})=\ldots=t(v_{i-l(i)})=1.$

We use $P(i,r,t)$ to denote the subpath of $P$ induced by nodes $v_1,v_2,\ldots,v_{i+r}$, where the threshold of each node $v_j$ with $j \neq i$ is $t(v_j)$, while the threshold of $v_i$ is set to  $ t \in \{t(v_i)-1, t(v_i)\}.$

We define $MIS[i,b,r,t]$ to be  the maximum number of nodes that can be influenced in $P(i,r,t)$ assuming that at most $b$ nodes among $v_1,v_2, \ldots,v_i$ belong to the target set while 
$v_{i+1}, \ldots,v_{i+r}$ do not.  

Noticing that $P(n,0,t(v_n))=P$ and we require that $|S|\leq \beta$, the desired value is $MIS[n,\beta,0,t(v_n)]$.

In order to get $MIS[n,\beta,0,t(v_n)]$, we compute $MIS[i,b,r,t]$ for each $i=0,1, \ldots n,$ for each $b=0,1,\ldots,\beta$, for each $r=0,1,\ldots,\min\{\lambda, r(i)\}$, and for $t \in \{t(v_i)-1, t(v_i)\}$.

Denote by $S(i,b,r,t)$ any target set attaining the value $MIS[i,b,r,t]$.

If $i=0$ OR $b=0$ we set $MIS[i,b,r,t]=0.$

If $i>0$ AND $b>0$. Consider the following quantities 
\begin{eqnarray*}
\ell &=& \min\{\lambda,l(i)\}\\
M_0 &=& 
\begin{cases} 
MIS[i{-}\ell{-}1,b-1,0,t(v_{i-\ell-1})-1] +r+\ell+1& \text{if $\ell<\lambda$}\\
MIS[i{-}\ell{-}1,b-1,0,t(v_{i-\ell-1})] +r+\ell+1 & \text{otherwise }
\end{cases}
\\
M_1&=& 
\begin{cases} 
MIS[i{-}1,b,0,t(v_{i-1})] & \text{if $t>1$}\\
MIS[i{-}1,b,\min\{\lambda,r+1\},t(v_{i-1})] & \text{otherwise.}
\end{cases}  
\end{eqnarray*}
By distinguishing whether  $v_i$ belongs to the target set $S(i,b,r,t)$ or not we are able to prove
that
\begin{equation*}\label{eq-path}
MIS[i,b,r,t]= \max \left\{M_0, M_1\right\}
\end{equation*}
and
$v_i \in S(i,b,r,t)$ if and only if  $MIS[i,b,r,t]= M_0$.
\end{proof}

For cycles, the problem can be solved by simply solving two different problems on a path and taking the minimum.
Indeed, starting with a cycle we can consider any node $v$ such that $t(v)\geq 2$ (if there is no such node, then the problem is trivial).
 If  node $v$ belongs to the target set, we can consider the path obtained by removing all the nodes influenced only by $v$ and then solve the problem on this path with a budget $\beta-1$. On the other hand, if we assume that $v$ does not belong to the target set,  then we simply consider the path obtained by eliminating $v$. 
 Therefore, we obtain the following result.
\begin{theorem}\label{teo-cycle}
The {$(\lambda, \beta)$-\textsc{Maximally Influencing Set}} problem can be solved in time $O(n\beta\lambda)$ on a cycle $C_n$.
\end{theorem}

%
%
%
%
%
%
%
%
%


Since complete graphs are of clique-width at most~$2$, results from \cite{CCGMV13} imply that the
$(\lambda,\beta)$-MIS problem is solvable in polynomial time on complete graphs if $\lambda$ is constant.
Indeed, one can see that for complete graphs the $(\lambda,\beta)$-\textsc{Maximally Influencing Set}  
can be solved in linear time, independently 
of the value of $\lambda$, by using ideas of  \cite{Nichterlein-12}.

If $G$ is a complete graph, we have that for any $S \subseteq V$, and any round $\tau \geq 1$,
it holds that
$$\Active[S, \tau] = \Active[S, \tau-1] \cup \{v \,:\,t(v) \leq |\Active[S, \tau-1]|\}.$$
Since $\Active[S, \tau - 1] \subseteq \Active[S, \tau],$ we have
\begin{equation} \label{eq:clique-dynamics}
\Active[S, \tau] = S \cup \{v \,:\,t(v) \leq |\Active[S, \tau-1]|\}.
\end{equation}
From (\ref{eq:clique-dynamics}), and by using a standard exchanging 
argument, one immediately sees that a set $S$
with largest influence is the one containing the nodes with highest
thresholds. Since $t(v)\in \{0, 1, \ldots, n\}$,
the selection of the $\beta$ nodes with highest threshold 
can be done in linear time. Summarizing, we have the following result.

\begin{theorem} \label{theorem:basic}
There exists an  optimal solution $S$
to the $(\lambda, \beta)$-\textsc{Maximally Influencing Set}  problem on a complete graph $G=(V, E),$ consisting of the
$\beta$ nodes of $V$ with highest thresholds, and it can be computed in linear time.
\end{theorem}

\section{Concluding Remarks}

We  considered the problems of selecting a \emph{bounded} cardinality
subset of people in (classes of)  networks, such that the influence they
spread, in a \emph{fixed}  number of rounds,  is the \emph{highest} among
all subsets of same bounded cardinality.
It is not difficult to see that our techniques can also solve closely related 
problems, in  the same classes of graphs considered in this paper.
 For instance, one could fix a requirement $\alpha$ and ask for
the \emph{minimum} cardinality target set such that after $\lambda$ rounds the number
of influenced people in the network is at least $\alpha$. Or, one could
fix a budget $\beta$ and a requirement $\alpha$, and ask about the \emph{minimum}
number $\lambda$ such that there exists a  target set of cardinality at most $\beta$
that influences at least $\alpha$ people in the network within $\lambda$ rounds
(such a minimum $\lambda$ could be equal to $\infty$).
Therefore, it is likely that the FONY\textsuperscript{\textregistered} Marketing Division
will have additional fun in solving these problems (and similar ones) as well.

%
%

\end{document}